\documentclass[preprint,review,10pt]{elsarticle}

\usepackage{amssymb}
\usepackage[hyphens]{url}
\usepackage[breaklinks=true]{hyperref}
\usepackage{pdflscape}
\usepackage{comment}

\usepackage{algorithm}
\usepackage{algorithmic}

\usepackage{lineno}

\usepackage{ctable}
\usepackage{tabularx}
\usepackage{pdflscape}

\newcolumntype{+}{>{\global \let \currentrowstyle \relax}}
\newcolumntype{^}{>{\currentrowstyle }}

\usepackage[cmex10]{amsmath}

\usepackage{booktabs}
\usepackage{multirow}

\newtheorem{theorem}{Theorem}[section]

\newenvironment{proof}{%
{\noindent \bf Proof. }%
}{%
\hfill$\Box$\\%
}

\usepackage{lscape}

\begin{document}

\begin{frontmatter}



\title{Neural Network-Enhanced Disease Spread Dynamics Over Time and Space}


\author[iitmath,mathbio]{Randy L. Caga-anan\corref{cor1}}

\address[iitmath]{Department of Mathematics and Statistics, MSU-Iligan Institute of Technology, Iligan City, Philippines}
\address[mathbio]{Mathematical Biology Research Group, Complex Systems Research Center, PRISM, MSU-Iligan Institute of Technology, Iligan City, Philippines}
\cortext[cor1]{Corresponding author: randy.caga-anan@g.msuiit.edu.ph}

\begin{abstract}
This study presents a neural network-enhanced approach to modeling disease spread dynamics over time and space. Neural networks are used to estimate time-varying parameters, with two calibration methods explored: Approximate Bayesian Computation (ABC) with Trust Region Reflective (TRF) optimization, and backpropagation with the Adam optimizer. Simulations show that the second method is faster for larger networks, while the first offers a greater diversity of acceptable solutions. The model is extended spatially by introducing a pathogen compartment, which diffuses through environmental transmission and interpersonal contact. We examine scenarios of exact reporting, overreporting, and underreporting, highlighting their effects on public behavior and infection peaks. Our results demonstrate that neural network-enhanced models more accurately capture dynamic changes in disease spread and illustrate the potential of scientific machine learning (SciML) to improve the predictive power of epidemiological models beyond traditional approaches.

\end{abstract}

\begin{keyword}neural networks \sep epidemiological modeling \sep scientific machine learning (SciML)

\MSC[2020] 92D30\sep 68T07  \sep 35Q92 \sep 37N25
\end{keyword}

\end{frontmatter}


\section{Introduction}

Scientific Machine Learning (SciML) emerges as a powerful framework that combines the strengths of traditional scientific computing and machine learning while addressing the limitations of each \cite{noordijk2024}. Traditional scientific computing techniques, such as ordinary differential equations (ODEs) and partial differential equations (PDEs), offer robust and interpretable models for capturing known dynamics, but struggle with uncertainties in parameters that are not directly measurable. On the other hand, machine learning, particularly artificial neural networks (ANNs), provides flexibility in approximating unknown relationships, yet it can lack interpretability and sometimes violate model assumptions. SciML bridges these gaps, providing a synergistic way to model complex systems by combining the best of both approaches.  

In the Philippines, the recent COVID-19 pandemic highlighted the importance of traditional modeling techniques and scientific computing in managing public health. Researchers developed and utilized ODE and PDE - based models to examine disease spread and intervention strategies \cite{arcede2020,caldwell2021, arcede2022}, such as vaccination programs \cite{cagaanan2021, campos2023} and quarantine policies \cite{macalisang2020, bock2020}. These models helped guide public health decisions by forecasting outbreaks and assessing the effectiveness of interventions, underscoring the relevance of traditional mathematical modeling in addressing real-world challenges.  

One major challenge in epidemic modeling is that the transmission rate is time-dependent, influenced by factors such as seasonality, policy shifts, and behavioral changes \cite{verelst2016}. Accurately capturing this variability is crucial for effective modeling. Additionally, the public response often shifts based on external information, such as news and social media, further complicating disease dynamics \cite{shahsavari2020}.  

In our modeling framework, we rely on ODEs and PDEs to represent the known structure of the disease spread. These methods capture the overall system behavior, including transitions between compartments like susceptible, infected, and recovered individuals. However, unknown time-varying parameters, such as the transmission rate and testing rate, present a challenge. To address these unknowns, we leverage neural networks to approximate the time-dependent parameters within the model.  

While physics-informed neural networks (PINNs) are widely used in recent research by embedding governing ODEs or PDEs within the loss function \cite{shaier2021, ning2023, millevoi2024, berkhahn2022}, this approach may present certain challenges. For instance, it does not guarantee that all compartment values will remain nonnegative over time, which could compromise the realism of the epidemiological model. To address these issues, we adopt a different strategy: employing traditional ODE-PDE frameworks as the model's foundation, while using neural networks to approximate unknown, time-dependent parameters. This hybrid approach ensures the model preserves the realistic epidemiological behaviors defined by the mechanistic model. By integrating neural networks with conventional scientific computing techniques in this way, we not only maintain realistic dynamics but also enhance the predictive capability of the epidemic model.

The rest of the paper is organized as follows. In Section \ref{secmodel}, we present the enhanced compartmental model and the architecture of the artificial neural network used. Section \ref{secparams} discusses the parameter values employed, while Section \ref{secalgo} introduces the two algorithms used to calibrate our model. In Section \ref{seccalib}, we provide the model calibration results and simulations. Section \ref{secrds} presents and simulates a spatial model using the calibrated parameters obtained from the previous sections. Finally, in Section \ref{secdc}, we offer our discussion and conclusion.

\section{Mathematical Model}\label{secmodel}
In \cite{cagaanan2023}, a modeling framework for the progression of COVID-19 was introduced and applied across the 17 regions of the Philippines, focusing on the optimal distribution of vaccines over time. In this study, we limit the application to the National Capital Region (NCR), the largest urban center and the most densely populated area in the country.

One advantage of the model’s structure is its ability to estimate untested infections, often referred to as the dark figure of the pandemic.

The population is divided into four compartments: Susceptible $(S_i)$, Infected $(I_i)$, Tested Positive $(T_i)$, and Removed $(R_i)$. Natural birth and death rates, along with reinfection, are excluded from the model, as the simulations in this study cover a relatively short period. Instead of introducing a separate compartment for vaccinated individuals, vaccination is incorporated by modeling the flow from the susceptible compartment to the removed compartment. The complete dynamics of the model are illustrated in Figure~\ref{figmodel}.
\begin{figure}[htbp]
\centering
\includegraphics[scale=0.5]{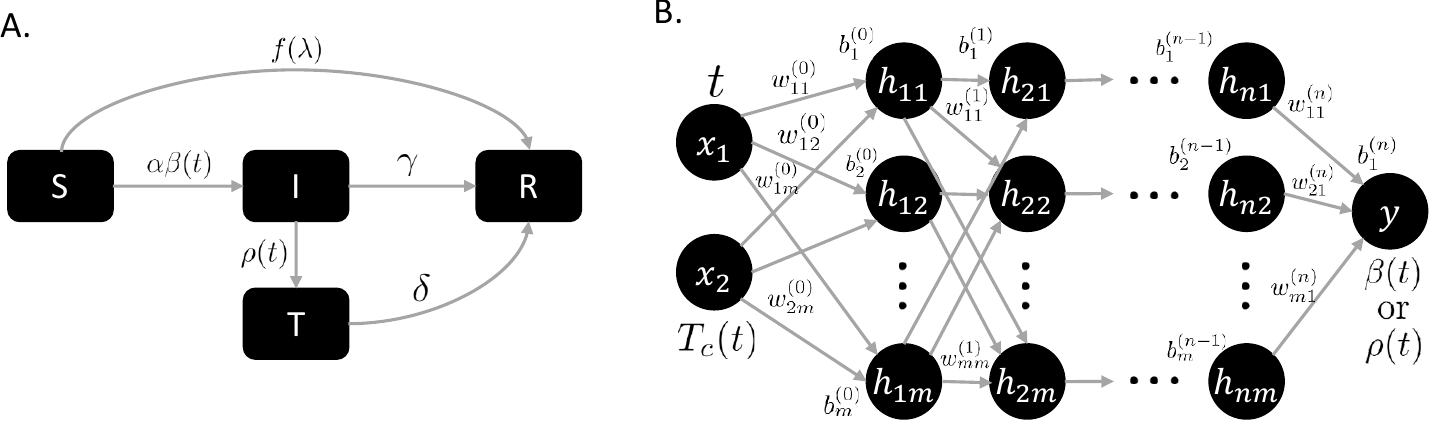} 
\\[-0.2cm]\caption{(A) Flowchart depicting the dynamics of the model, where the population is categorized into Susceptible, Infected, Tested Positive, and Removed compartments. (B) Architecture of the artificial neural network: a multi-layer perceptron with two input neurons, $n$ hidden layers each containing $m$ neurons, and one output neuron. The hidden layers use ReLU as the activation function, while the output layer employs the sigmoid activation function.}
\label{figmodel}
\end{figure}

The model is governed by the following system of differential equations:
\begin{equation} \label{odemodel}
\begin{aligned}
\frac{dS}{dt} &=-\alpha\beta(t) I \frac{S}{N_0} - f(\lambda)\\
\frac{dI}{dt} &= \alpha\beta(t) I \frac{S}{N_0} - (\gamma + \rho(t))I \\
\frac{dT}{dt} &= \rho(t) I - \delta T \\
\frac{dR}{dt} &= \gamma I + \delta T + f(\lambda), 
\end{aligned}   
\end{equation}
where $\alpha, \gamma, \delta >0$; $\lambda$ is a positive integer; $0\leq \beta(t), \rho(t)\leq 1$, for all $t\geq 0$; and $N_0$ is the total population at $t=0$. Moreover, $f(\lambda) = f_1(\lambda) = \lambda$ or $f(\lambda) = f_2(\lambda) = \lambda\frac{S}{N_0}$. The descriptions of the parameters are given in Table~\ref{tab:parameters}. 

When using $f_1$ for $f$, the time domain is limited to $t_f$, ensuring that $\lambda < S(t_f)$, so the daily number of vaccinated individuals does not exceed the number of susceptibles. This model is the same as the one used in \cite{cagaanan2023}. Note that at the onset of the infection, when almost everyone is susceptible, $f_1\approx f_2$.

We will primarily use $f_1$ for the calibration of model (\ref{odemodel}) (Section \ref{seccalib}), to enable comparison with the results in \cite{cagaanan2023}. In Section \ref{secrds}, we will employ $f_2$ for our spatial model. The vaccination rate $f_1$ is not suitable for that model because its value is designed for the entire population and is not directly transferable as a vaccination rate for a spatial unit within the model.

Using $f_2$ for $f$, the following theorem tells us the long-term dynamics of the model.
\begin{theorem}\label{globalstab}
The model described in (\ref{odemodel}), with $f = f_2$, has a disease-free equilibrium at $(0, 0, 0, N_0)$, which is globally asymptotically stable.
\end{theorem}
\begin{proof}
Observe that $\frac{dS}{dt}+\frac{dI}{dt}+\frac{dT}{dt}+\frac{dR}{dt}=0$. Hence, the sum of the compartments remains unchanged over time. It is then easy to see that the disease-free equilibrium $(S^*,I^*,T^*,R^*)$ is given by $(0,0,0,N_0)$.
Consider that 
\begin{equation}
	R(t)-R(0)=\int_{0}^{t}\left(\gamma I(s) + \delta T(s) + \lambda\frac{S(s)}{N_0}\right)\, ds
\end{equation}
and 
\begin{equation}
	R^*-R(0)=\int_{0}^{+\infty}\left(\gamma I(s) + \delta T(s) + \lambda\frac{S(s)}{N_0}\right)\, ds.
\end{equation}
This is finite. Hence, $\gamma I(t) + \delta T(t) + \lambda\frac{S(t)}{N_0} \to 0$ as  $t\to +\infty$.
Since each term is nonnegative, we thus have $S(t),\, I(t),\, T(t) \to 0$ as $t\to +\infty$.
\end{proof}

In Figure \ref{figmodel} (B), we present our artificial neural network architecture with 2 input neurons, $n$ hidden layers each containing $m$ neurons, and one output neuron. It is a multi-layer perceptron, or a network where succeeding layers are fully connected. For each arrow in the network, we associate one weight parameter, and for each neuron in the hidden and output layers, we associate a bias parameter. Thus, the number of learnable parameters is $2m + (n - 1)m^2 + m$ (for the weights) plus $nm + 1$ (for the biases), which equals 
\begin{equation}\label{totalwb}
(n - 1)m^2 + (3 + n)m + 1.
\end{equation}
The hidden layers use the rectified linear unit (ReLU) ($f(x) = \max(0, x)$) as the activation function, while the output layer employs the sigmoid activation function, $\sigma(x) = \frac{1}{1 + e^{-x}}$. The computation from the input layer to the output layer follows the standard procedure: calculate the weighted sum of inputs from the preceding layer, add the bias, and apply the activation function to introduce non-linearity. Repeat these steps successively from the hidden layers to the output layer. In our model, the two inputs are time $t$ and the cumulative number of tested positive cases at time $t$, denoted by $T_c(t)$. The output is the value of the parameter $\beta$ or $\rho$ at time $t$.

\section{Parameter Values}\label{secparams}
The confirmed cases data for NCR used in parameter calibration spans the period from June 1, 2021, to October 31, 2021, and was made publicly available by the Philippines' Department of Health (DOH) \cite{dohtracker}. This is the same dataset used in \cite{cagaanan2023} to determine the fitted values of the parameters $\beta$, $\alpha$, $\rho$, and $I_0$ using the Approximate Bayesian Computation approach combined with the Levenberg-Marquardt algorithm \cite{csillery2010, levenberg1944, marquardt1963}.

We will employ the same technique as one of the methods to calibrate the weights and biases, as explained in the next section. In this study, we fix the parameters and initial conditions based on \cite{cagaanan2023} but calibrate the weights and biases of the neural networks for $\beta$ and $\rho$. Please refer to Table~\ref{tab:parameters} for the values. In some simulations, we focus only on the neural network for $\beta$, in which case the value of $\rho$ is fixed as reported in \cite{cagaanan2023}.

Notably, \cite{cagaanan2023} reported a relative error of $4.7 \times 10^{-4}$ with constant values for $\beta$ and $\rho$. In Section \ref{seccalib}, we will demonstrate that this error can be reduced by optimizing the neural networks for $\beta$ and $\rho$.

\begin{table}[htbp]
	\centering
\begin{tabular}{c|c|c}
	\hline 
\textbf{Symbol}	& \textbf{Description} & \textbf{Value} \\ 
        \hline 
	\hline 
	$\lambda$ & vaccination rate &    37430 humans/day \\
	\hline 
	$\beta$ & time-dependent transmission rate factor & neural network \\ \hline 
	$\alpha$ & transmission rate constant &  0.44696 /day \\
	\hline 
	$\rho$ &  detection rate &  0.01187 /day \\ 
	\hline 
	$\delta$ &  removal rate from $T$ to $R$ & 0.1428 /day  \\ 
	\hline 
	$\gamma$ &  removal rate from $I$ to $R$  & 0.0833 /day  \\ 
	\hline 
        \hline 
	$N_0$ &  total population  & 13484462 humans  \\ 
        \hline 
	$S_0$ &  initial value for $S$  & 11381077 humans  \\ 
        \hline 
	$I_0$ &  initial value for $I$  & 12888 humans \\ 
	\hline 
	$T_0$ &  initial value for $T$  & 7750 humans \\ 
        \hline 
	$R_0$ &  initial value for $R$  & 522624 humans \\ 
	\hline 
\end{tabular} 
\caption{Parameters and Initial Conditions. The parameter values and initial conditions are taken from \cite{cagaanan2023} for NCR. In \cite{cagaanan2023}, the parameter $\beta$ was fitted to 0.43619. With this set of parameters, the paper reported a relative error of 0.00047.}
\label{tab:parameters}
\end{table}

\section{Algorithms}\label{secalgo}
We use two approaches to optimize the weights and biases of our neural networks. The first approach is similar to the optimization method employed in \cite{cagaanan2023}, which utilizes the Approximate Bayesian Computation (ABC) method. In this study, however, we combine it with the Trust Region Reflective (TRF) minimization algorithm \cite{branch1999, byrd1988}. TRF effectively handles bounded optimization problems and is well-suited for large-scale problems with many parameters or data points. Please refer to Algorithm 1 for the implementation details.

The second approach employs machine learning techniques. Specifically, we trained the weights and biases using backpropagation with the Adaptive Moment Estimation (Adam) optimizer \cite{kingma2014}. Refer to Algorithm 2 for the implementation details. 

For Algorithm 1, the cost function to be minimized is the nonlinear least squares function, defined as 
$$L_1(\theta) = \sum_{t=1}^{n}(T_{d}(t_i) - T_c(t_i))^2,$$
where $n$ is the number of data points, $\theta$ is the set of parameters to be fitted, and $T_{d}(t_i)$ and $T_c(t_i)$ represent the cumulative confirmed cases from the data and the model output at time $t_i$, respectively. For Algorithm 2, the loss function is the mean squared error (MSE), given by $$L_2 = L_1/n.$$ The relative error $r_e$ is calculated by dividing $L_1$ by the sum of the squares of the data values.

To use Algorithm 1, we must set bounds for the weights and biases. In our case, we restrict them to the interval $(-3, 3)$. No bounds are required for Algorithm 2.

In both algorithms, the predicted values $y_p$ correspond to the cumulative tested positive cases $T_c$, which is the output of the enhanced model using $\beta(t, T_c(t))$ (and $\rho(t, T_c(t))$) for each $t$. 

We implement the forward pass in Algorithm 1 and the training routine of Algorithm 2 using PyTorch \cite{paszke2019}.

\begin{algorithm}\label{algoabc}
\caption{ABC with TRF}
\begin{algorithmic}[1] 
\STATE \textbf{Input:} Observed data $y$, tolerance $\epsilon$, number of samples $N$, prior distribution $P(\theta)$, loss function $L_1(y,y_p)$, parameters feasible space $\Delta$
    \FOR{$i = 1$ to $N$}
            \STATE Sample a set of parameters $\theta_0$ from $P(\theta)$.
            \STATE Solve for $$\theta_{opt}=\min_{\theta\in\Delta} L_1(y,y_p(\theta)),$$
            with $\theta_0$ used as the initial guess for the minimization method (TRF) and $y_p(\theta)$ is the model output using the set of parameters $\theta$.
            \IF{$L_1(y,y_p(\theta_{opt})) < \epsilon$}
                \STATE Accept $\theta_{opt}$. Add to the set $\Theta_{accepted}$.
            \ELSE
                \STATE Reject $\theta_{opt}$.
            \ENDIF
    \ENDFOR
 \STATE \textbf{Output:} Set of accepted parameters $\Theta_{accepted}$
\end{algorithmic}
\end{algorithm}

\begin{algorithm}\label{algoml}
\caption{Training with Adam Optimizer and MSE Loss}
\begin{algorithmic}[1] 
\STATE \textbf{Input:} Training data $y$, learning rate $\epsilon$, initial weights $W_0$  and biases $b_0$, number of epochs $N$

    \FOR{$i = 1$ to $N$}
            \STATE \textbf{Forward Pass:} Compute the model output $y_p(W,b)$ using the current  weights $W$ and biases $b$. Initially, use $W_0$ and $b_0$ as the starting values.
            \STATE \textbf{Compute Loss (MSE):} Calculate $L_2(y,y_p)$.
            \STATE \textbf{Compute Gradients:} Calculate $\frac{\partial L_f}{\partial W}$ and $\frac{\partial L_f}{\partial b}$ using backpropagation.
            \STATE \textbf{Apply Adam Optimizer:} Update weights $W$ and biases $b$.
    \ENDFOR
 \STATE Save the final weights and biases as $W_{opt}$ and $b_{opt}$, respectively.
 \STATE \textbf{Output:} Optimized weights $W_{opt}$ and biases $b_{opt}$
\end{algorithmic}
\end{algorithm}

\section{Model Calibration and Simulations}\label{seccalib}
We consider six cases summarized in Table \ref{tab:results}. The first four cases involve model (\ref{odemodel}) with $\beta(t)$ modeled as a neural network while keeping $\rho$ fixed. The last two cases use model (\ref{odemodel}) with both $\beta(t)$ and $\rho(t)$ as neural networks. We alternate between Algorithms 1 and 2 to compare their performance across three scenarios.

First, we evaluate a simple architecture for $\beta(t)$ with 2 hidden layers, each containing 2 neurons. This configuration results in a total of 15 learnable parameters, as computed from (\ref{totalwb}). Next, we examine a more complex case by increasing the number of hidden layers to 3, each containing 4 neurons, which significantly increases the learnable parameters to 57. Finally, we consider a scenario where both $\beta(t)$ and $\rho(t)$ are represented as neural networks, each using the simpler architecture of 2 hidden layers with 2 neurons per layer.

For Algorithm 1, we set the tolerance $\epsilon = 10^{-3}$. With 2000 samples for Algorithm 1 and 2000 epochs for Algorithm 2, we can already find parameters that yield a relative error smaller than that reported in \cite{cagaanan2023} for all cases except the last. In the last case, after 2000 epochs, the relative error remains slightly higher than the reported value in \cite{cagaanan2023}. However, with 4000 epochs, the error is  now slightly better, and after 6000 epochs, it achieves a relative error on the order of $10^{-6}$, which is significantly lower than the one obtained in \cite{cagaanan2023}.

The runtime column in Table \ref{tab:results} highlights that the machine learning approach is significantly faster, particularly as the network complexity increases (3.121 hours compared to 8.615 hours). Although Algorithm 1 accepts multiple parameter sets, only the set with the minimum relative error is reported in the relative error column of Table \ref{tab:results}.

Figure \ref{abcvsml1} shows the outputs for the first two cases in Table \ref{tab:results}, while Figure \ref{abcvsml2} presents the outputs for the last two cases. In each figure, the left column shows the results obtained using the machine learning approach, while the right column shows those from the ABC with TRF method. It is evident that the ABC with TRF method yields a more diverse set of accepted parameters. However, the machine learning approach achieves the best fit.

\begin{table}[htbp]
	\centering
\begin{tabular}{c|c|c|c}
	\hline 
\textbf{Algorithm} & \textbf{Neural Net}	& \textbf{Relative Error} & \textbf{Runtime} \\ 
	\hline 
        \hline
	1 (2000 samples) & $\beta:2h2n:15p$ &  $8.29\times 10^{-5}$ & 1.344 hours \\
	\hline 
	2 (2000 epochs) & $\beta:2h2n:15p$ &  $1.29\times 10^{-4}$ & 2.823 hours \\
	\hline 
        1 (2000 samples) & $\beta:3h4n:57p$ &  $7.98\times 10^{-5}$ & 8.615 hours \\
	\hline 
	2 (2000 epochs) & $\beta:3h4n:57p$ &  $1.28\times 10^{-4}$ & 3.121 hours \\
	\hline 
        1 (2000 samples) & $\beta,\rho:2h2n:30p$ &  $1.96\times 10^{-4}$ & 6.791 hours \\
	\hline 
	2 (6000 epochs) & $\beta,\rho:2h2n:30p$ &  $7.44\times 10^{-6}$
 & 8.462
 hours \\
	\hline 
\end{tabular} 
\caption{Calibration cases considered for model (\ref{odemodel}).}
\label{tab:results}
\end{table}

\begin{figure}[htbp]
\centering
\includegraphics[scale=0.37]{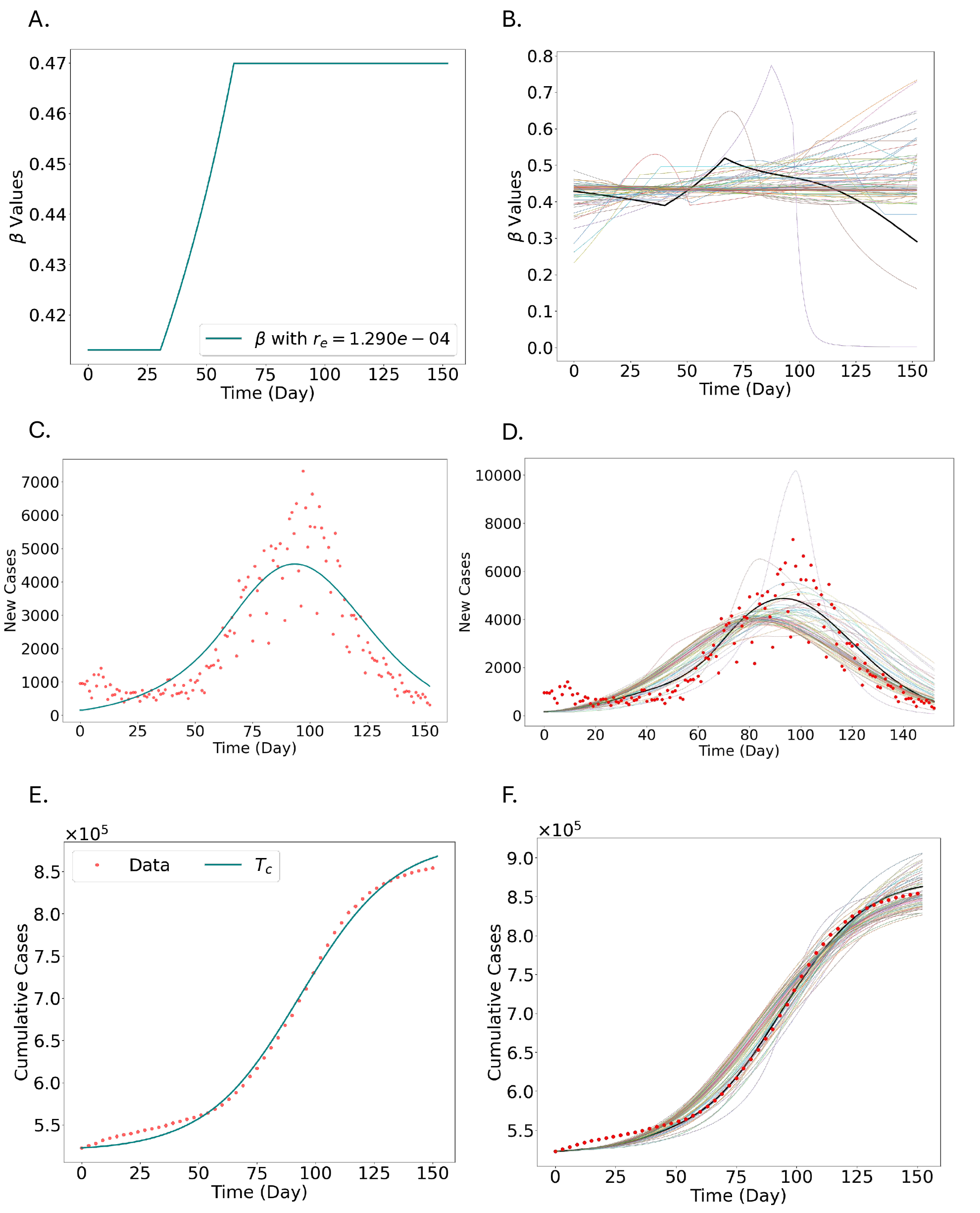} 
\\[-0.2cm]\caption{Model calibration outputs for cases 1 and 2 from Table \ref{tab:results}. The left column (A., C., E.) presents the output from the machine learning method (Algorithm 2), while the right column (B., D., F.) shows the output from the ABC with TRF approach (Algorithm 1). In the right column, the curves represent the accepted parameter sets ($L_1(y, y_p(\theta_{opt})) < \epsilon$), with the highlighted curve (in black) indicating the set with the minimum relative error. The first row (A., B.) displays the $\beta(t)$ values, the second row (C., D.) shows the fit for new confirmed cases, and the third row (E., F.) presents the fit for cumulative confirmed cases.}
\label{abcvsml1}
\end{figure}

\begin{figure}[htbp]
\centering
\includegraphics[scale=0.38]{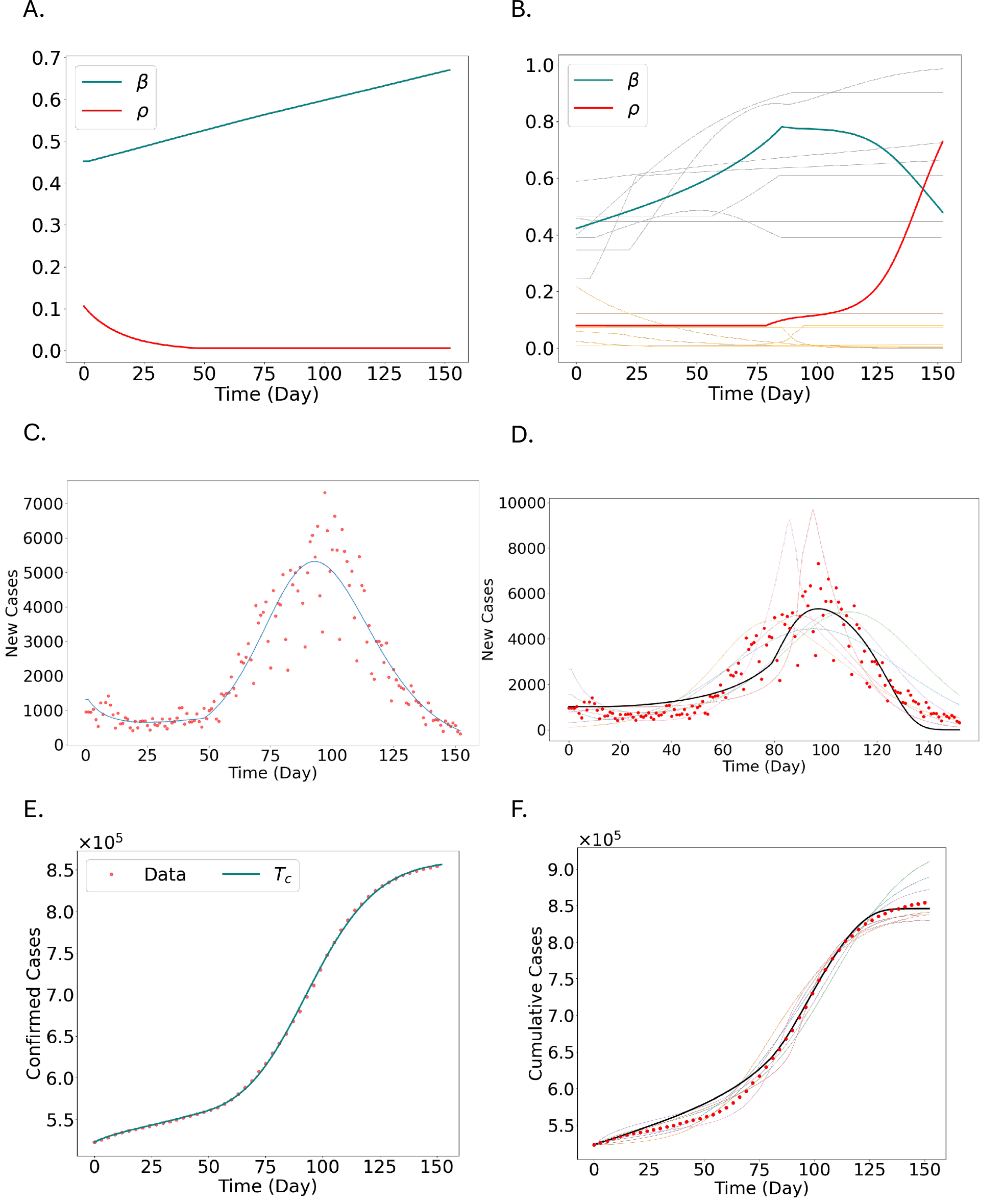} 
\\[-0.2cm]\caption{Model calibration outputs for the last two cases from Table \ref{tab:results}. In these cases, both $\beta(t)$ and $\rho(t)$ are modeled as neural networks. The descriptions are the same as those in Figure \ref{abcvsml1}, except that in the first row (A., B.), we now have both $\beta(t)$ and $\rho(t)$.}
\label{abcvsml2}
\end{figure}


\section{Enhanced Reaction-Diffusion System}\label{secrds}
In this section, we will use the calibrated neural networks for $\beta(t)$ and $\rho(t)$ for model (\ref{odemodel}) but with a space variable. We observe that, in general, people do not diffuse over space but return to their homes every day, so we do not include diffusion for the human compartments $S$, $I$, $T$, and $R$. Instead, we introduce a new variable $P$, representing the pathogen—in our case, the COVID-19 virus—which can diffuse over space by being transmitted from person to person or through the environment.

The dynamics are governed by a system of five differential equations as follows, for $t > 0$ and $x\in \Omega\subset\mathbb{R}^2$,
\begin{equation}\label{pdemodel}
\begin{aligned}
    \frac{\partial S(t, x)}{\partial t} &= - \frac{\alpha \beta(t) S(t, x) P(t, x)}{N(t, x)} - \frac{\Lambda S(t, x)}{N_0} \\
    \frac{\partial I(t, x)}{\partial t} &= \frac{\alpha \beta(t) S(t, x) P(t, x)}{N(t, x)} - (\gamma + \rho(t) ) I(t, x) \\
    \frac{\partial T(t, x)}{\partial t} &= \rho(t)\, I(t, x) - \delta\, T(t, x) \\
    \frac{\partial R(t, x)}{\partial t} &= \gamma\, I(t, x) + \delta\, T(t, x) + \frac{\Lambda S(t, x)}{N_0} \\
    \frac{\partial P(t, x)}{\partial t} &= D \nabla^2 P(t, x) + \xi_I I(t, x) + \xi_T T(t, x) - \zeta P(t, x),
\end{aligned}
\end{equation}
where $N(t, x) = S(t, x) + I(t, x) + T(t, x) + R(t, x)$, $D$ is the diffusion coefficient, and $\nabla^2 P(t, x)$ is the Laplacian (spatial diffusion term) that models how the pathogen spreads across space. Here, we imposed homogeneous Neumann boundary conditions.

We introduce additional parameters: $\xi_I$ and $\xi_T$, which represent the rates of quanta generation by infected ($I$) and tested positive ($T$) individuals, respectively. A quantum of pathogen refers to a unit of pathogen concentration sufficient to cause infection in a susceptible individual. Finally, $\zeta$ denotes the pathogen clearance rate.

Figure \ref{figpdeinitial} presents the initial states ($t = 0$) of $S$, $I$, and $P$ over the space domain, in our case the NCR region. The population density for $S$ is taken from \cite{citypopulation}. Since no publicly available spatial data exists for confirmed cases in NCR, we generated the states for $I$ and $P$ for illustrative purposes.

For the fixed parameter values, we used those from Table \ref{tab:parameters}. For the additional parameters, we estimated the following values: $\xi_I = 0.85$, $\xi_T = 0.25$, $\zeta = 0.75$, and $D = 140.0625$. The value for $D$ is from \cite{arcede2022}, but it was multiplied by 10 because the original value represents the average for the whole Philippines, whereas NCR has the highest level of human activity, allowing the pathogen to spread more rapidly.

We numerically solve model (\ref{pdemodel}) using the finite volume method. Figure \ref{figpdeevolution} shows the time and space evolution of $I$ and $P$. We simulated the model for 152 days, but here we only present the states at $t = 30, 90,$ and $150$.

As noted in the introduction, public response can be influenced by overreporting or underreporting in the news and social media. For example, overhyped cases may lead the public to adhere more strictly to protocols, while underreporting may cause them to become complacent. These behaviors, in turn, affect the subsequent number of cases, which then influence the next public response, creating a cyclical pattern. We use our calibrated $\beta$ and $\rho$ to simulate these scenarios. For exact reporting, we use $\beta(t, T_c(t))$ and $\rho(t, T_c(t))$. In the case of overreporting, we use $\beta(t, 2T_c(t))$ and $\rho(t, 2T_c(t))$, while for underreporting, we use $\beta(t, 0.5T_c(t))$ and $\rho(t, 0.5T_c(t))$. The results are shown in Figure \ref{figpdereunderoverreporting}. One may notice that Figure \ref{figpdereunderoverreporting} (A) shows double peaks, enabled by our time-dependent parameters $\beta$ and $\rho$. This phenomenon, which may be more realistic, cannot occur if the parameters are constant over time.

\begin{figure}[htbp]
\centering
\includegraphics[scale=0.4]{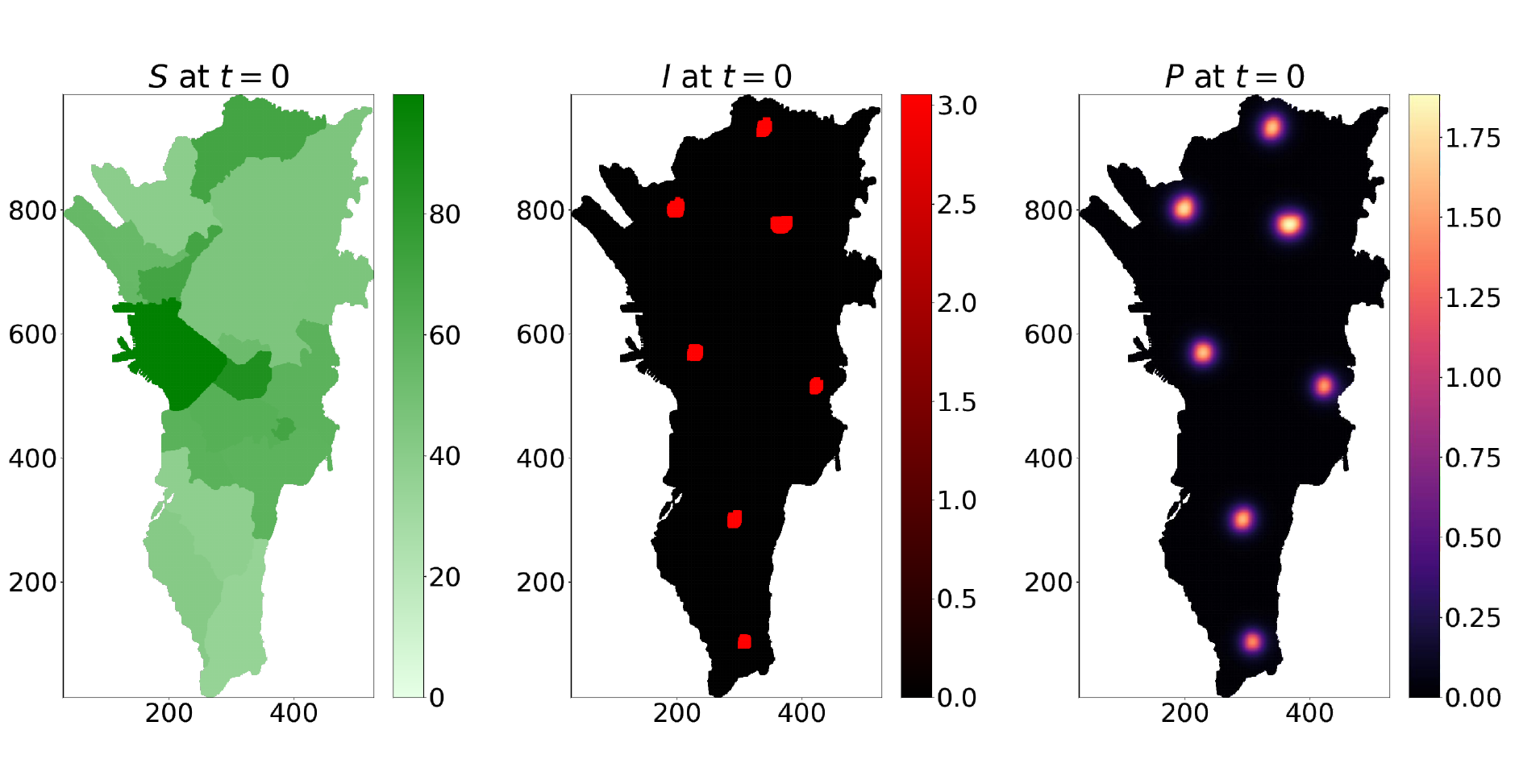} 
\\[-0.2cm]\caption{Initial states ($t=0$) used in the simulations.}
\label{figpdeinitial}
\end{figure}

\begin{figure}[htbp]
\centering
\includegraphics[scale=0.4]{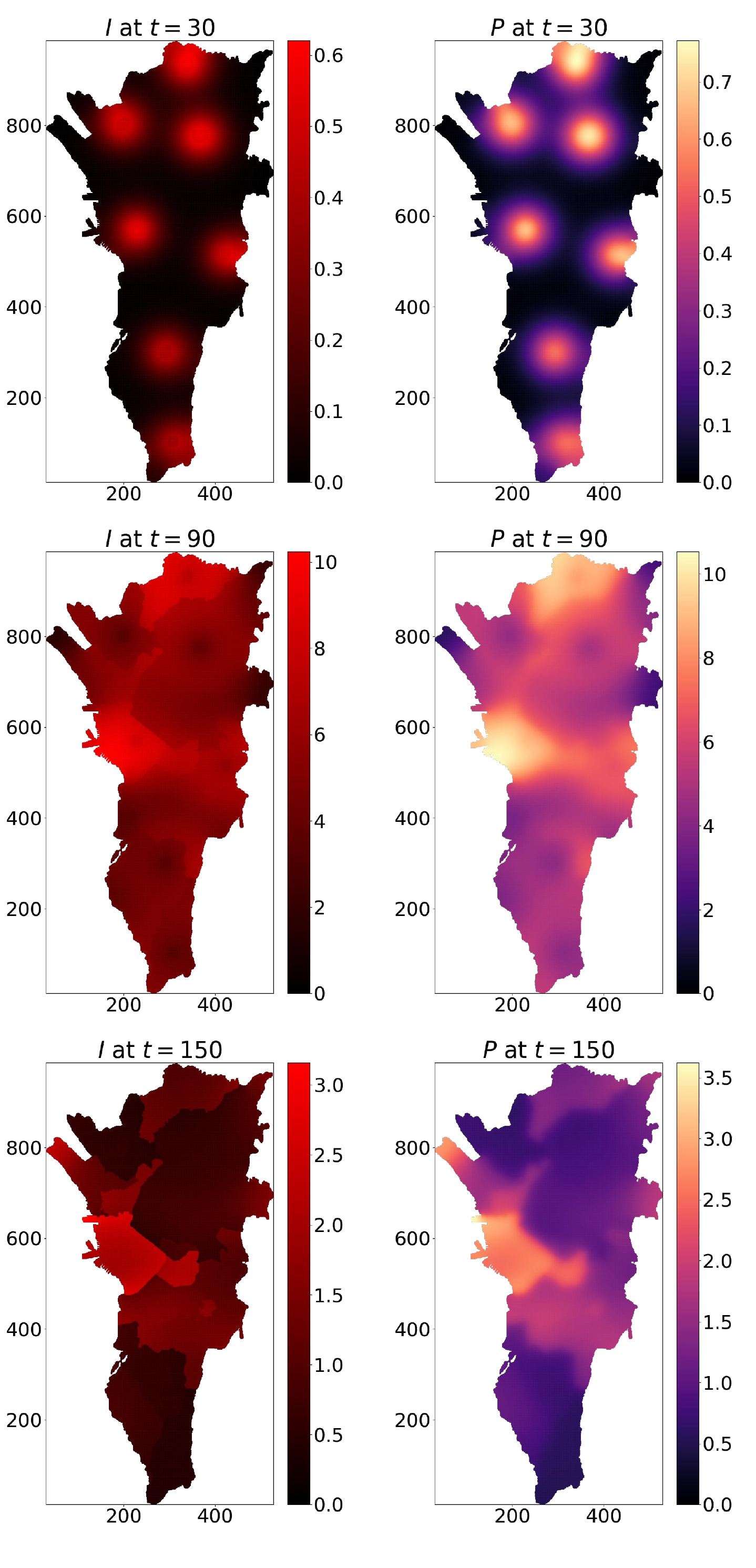} 
\\[-0.2cm]\caption{Time and space evolution of $I(t, x)$ and $P(t, x)$.}
\label{figpdeevolution}
\end{figure}

\begin{figure}[htbp]
\centering
\includegraphics[scale=0.5]{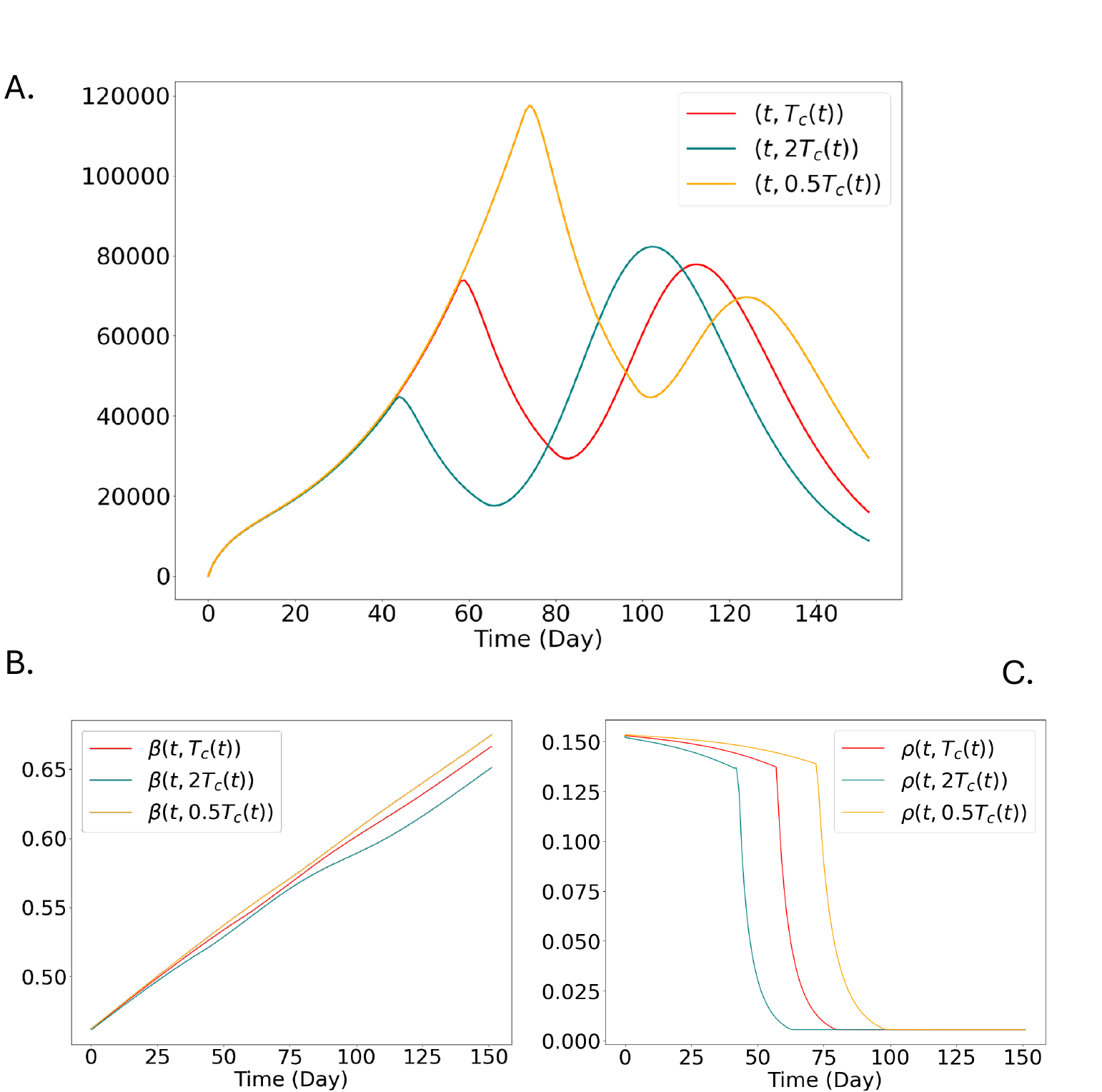} 
\\[-0.2cm]\caption{Simulations of the effect of overreporting and underreporting using the calibrated neural networks for $\beta$ and $\rho$. (A) $T(t) = \int_\Omega T(t, x)\,dx$ for the three cases: true, under, and overreporting. (B) $\beta$ for the three cases. (C) $\rho$ for the three cases.}
\label{figpdereunderoverreporting}
\end{figure}

\section{Discussion and Conclusion}\label{secdc}
In \cite{cagaanan2021}, we defined time-dependent parameters as piecewise constant functions, manually setting the times at which the parameters change their values. In this study, we demonstrated that we can leverage machine learning techniques to estimate both when and how parameters change over time. For example, in Figure \ref{abcvsml1} (A), our trained neural network revealed that the transmission rate $\beta$ began to change after day 25, with a new, increased constant rate emerging around day 60—entirely based on the data.

Considering our two calibration algorithms, if the goal is to obtain diverse, acceptable parameter sets with shallow neural networks, Algorithm 1 would likely be the best choice. However, if speed is required, especially when working with large or deep neural networks, Algorithm 2, which is designed for such cases, would be the better option.

In Section \ref{secrds}, we applied our calibrated neural networks for $\beta$ and $\rho$ to the spatial version of model (\ref{odemodel}). A key feature of this model is the introduction of a pathogen compartment, where diffusion occurs. This adaptation reflects the fact that, generally, populations do not diffuse in the traditional sense within the spatial domain, as people tend to return home daily.

Figure \ref{figpdereunderoverreporting} (A) illustrates the potential effects of overreporting and underreporting in news or social media. If overreporting occurs, people may follow protocols more strictly, resulting in a lower first peak. However, this lower peak could lead to complacency, causing a higher second peak. The reverse can happen with underreporting: an initially higher peak, followed by a smaller second peak due to increased vigilance afterward.

In summary, we demonstrated one possible application of scientific machine learning (SciML) by improving epidemiological models through the use of neural networks to represent time-dependent parameters. Our simulations showed significant improvements over traditional models.






\bibliographystyle{plain}

\end{document}